\documentclass[letterpaper,10pt,conference]{ieeeconf}  
\IEEEoverridecommandlockouts                              
\usepackage{fancyhdr}
\usepackage{amssymb}
\usepackage{amsmath}
\usepackage{graphicx}
\usepackage{listings}
\usepackage{float}
\usepackage{enumerate}
\usepackage{epstopdf}
\usepackage{color}
\usepackage{setspace}
\usepackage{bm}
\usepackage{bbm}
\usepackage{cite}
\usepackage{setspace}
\usepackage{cite}
\usepackage{dsfont}
\usepackage{tabu}
\usepackage[mathscr]{euscript}
\usepackage{threeparttable}
\usepackage{multirow}
\usepackage{mathtools}
\usepackage{xcolor}
\usepackage{soul}

\include{hollyMacros}

\newcommand{\aaa}{\mathcal{A}}

\DeclareMathOperator*{\argmax}{arg\,max}
\DeclareMathOperator*{\argmin}{arg\,min}

\newcommand{\vs}{\vspace{-1mm}}
\newcommand{\hs}[1][0.5]{\hspace{-#1mm}}

\newcommand{\aem}{a^{\rm e}}

\newcommand{\A}{{\rm A}}
\newcommand{\M}{{\rm M}}

\newcommand{\K}{\mathcal{K}}
\newcommand{\R}{\mathbb{R}}

\DeclarePairedDelimiter{\ceil}{\lceil}{\rceil}
\DeclarePairedDelimiter{\floor}{\lfloor}{\rfloor}

\newtheorem{theorem}{Theorem}
\newtheorem{defn}{Definition}
\newtheorem{lemma}{Lemma}

\normalsize\title{\LARGE \bf
Exploiting an Adversary's Intentions in Graphical Coordination Games}

\author{
Brandon C. Collins and
Philip N. Brown
\thanks{Funding support provided by Colorado State Bill 18-086, the UCCS Undergraduate Research Academy, and the UCCS Committee on Research and Creative Works.}
\thanks{Brandon C. Collins (corresponding author) and Philip N. Brown are with the Department of Computer Science at the University of Colorado Colorado Springs, {\tt \{bcollin3,philip.brown\}@uccs.edu}.}
}
\begin{document}

\maketitle

\begin{abstract}
    How does information regarding an adversary's intentions affect optimal system design?
    This paper addresses this question in the context of graphical coordination games where an adversary can indirectly influence the behavior of agents by modifying their payoffs.
    We study a situation in which a system operator must select a graph topology in anticipation of the action of an unknown adversary.
    The designer can limit her worst-case losses by playing a \emph{security strategy}, effectively planning for an adversary which intends maximum harm.
    However, fine-grained information regarding the adversary's intention may help the system operator to fine-tune the defenses and obtain better system performance.
    In a simple model of adversarial behavior, this paper asks how much a system operator can gain by fine-tuning a defense for known adversarial intent.
    We find that if the adversary is weak, a security strategy is approximately optimal for any adversary type; however, for moderately-strong adversaries, security strategies are far from optimal.
\end{abstract}

\graphicspath{{figures/}}

\section{Introduction}

Many engineered and social systems consist of a collection of agents making decisions based on locally available information.
For example, a group of unmanned vehicles performing surveillance in a hostile area may use a distributed control strategy due to communication constraints;
social systems are intrinsically distributed, as individuals make decisions based on personal objectives and in response to the behavior of friends and acquaintances.
For example, the decision to adopt a recently released technology may depend both on the quality of the item itself and on friends' choices~\cite{Kreindler2014}.

In a multiagent system, whether engineered or social, an adversary may be able to influence the behavior of individual agents to indirectly influence the overall behavior of the system.
In principle, a defender can compute and deploy a \emph{security strategy}, which effectively plans for a maximally-malicious adversary.
Security strategies allow the defender to guarantee a minimum level of system performance, but as worst-case defenses, they may be suboptimal against a real adversary.
In this paper, we investigate the cost of employing security strategies against poorly-characterized adversaries, in particular studying how information regarding an adversary's \emph{intent} and \emph{strength} can be used to fine-tune defensive measures.

A popular theoretical framework in this area is to model the agents in a multiagent system as players in a game, endow the players with a distributed learning algorithm (such as the well-studied log-linear learning rule~\cite{Blume1993, Shah2010}), and analyze the emergent behavior resulting when the agents execute their algorithm.
%
%
This has prompted much recent interest the question: how can adversarial manipulation  alter the emergent behavior of distributed decisionmaking~\cite{Borowski2015,Brown2017f,Canty2018,Paarporn2019}?

Work in this area has characterized how an adversary's capabilities and intelligence impacted their ability to influence behavior~\cite{Brown2017f,Canty2018}, as well as investigated simple questions of {defense}~\cite{Paarporn2019}.

In this paper, we ask how a system operator can leverage information about an adversary's \emph{goals} (or intent) to select an effective defensive posture.
Here, we suppose that the system operator must select a defensive posture without knowing which type of adversary she will face.
That is, a central question of this paper is this: how should a system operator implement effective defenses if the \emph{goals} of the adversary are unknown?
Naturally, one option is that the planner could adopt a pessimistic view and plan for the worst-case adversary type; however, we demonstrate that doing so can be quite costly in some circumstances.

We study two models of adversary intent:
\begin{itemize}
\item A \emph{malicious} adversary wishes to harm the system operator as much as possible, and
\item An \emph{advertiser} adversary has a vested interest in some specific agent action.
\end{itemize}
%
%
We investigate the system operator's defensive decision problem as affected by the adversary's strength (measured by how many friendly agents the adversary is able to influence) and the adversary's goals.
We find that there are distinct regimes of adversary strength:
when the adversary is ``weak'' (only able to influence a small set of agents), a security strategy is nearly optimal.
On the other hand, if the adversary is moderately powerful, then security strategies are always far from optimal, and can only garner a maximum of a fraction of the total available welfare for an advertiser.

\section{Model}\label{sec:model}

\subsection{Model of agent and adversarial behavior}
To study adversarial influence we use the $n$-player \emph{graphical coordination game} \cite{Kearns2001,Young2011, Montanari2010}.
The foundation of a graphical coordination game is a simple two agent coordination game\cite{Cooper1999, Ullmann-Margalit1977}, where each agent must choose one of two conventions, $\{x,y\}$, with payoffs depicted by the following payoff matrix\eqref{eq:matrix}: \vspace{-4mm}

\begin{equation}
\setlength{\extrarowheight}{2pt}
\mbox{
\begin{tabular}{r|c|c|}
\multicolumn{1}{r}{}	&\multicolumn{1}{c}{$x$}	&\multicolumn{1}{c}{$y$}\\
\cline{2-3}$x$			&$1+\alpha,\,1+\alpha$	&$0,\,0$\\
\cline{2-3}$y$			&$0,\,0$				&$1,\,1$\\
\cline{2-3}
\end{tabular}}\label{eq:matrix}
\end{equation}

The game is played on undirected graph $g=(N,E)$ between agents $N=\{1,2,3,...,n\}$ using edge set $E$. 
Agent $i$ plays the two-player coordination game with agent $j$ if $(i,j)\in E$ and $i\neq j$. 
We term $\alpha\in(0,1]$ as the \emph{payoff gain} that ensures the $x$ convention has an intrinsic coordination benefit. 
An agent's total payoff is the sum of payoffs it receives in the two-player games played with its neighbors ${\cal N}_i = \{j \in N : (i,j) \in E\}$, i.e., for a joint action $a = (a_1, \dots, a_n) \in \aaa := \{x,y\}^n$, the utility function of agent $i$ is
\begin{equation}\label{e:original utility}
U_i(a_1, \dots, a_n) = \sum_{j \in {\cal N}_i} u(a_i, a_j) - \left\{
\begin{array}{ll}
0 & \mbox{if $a_i = x$} \\
p & \mbox{if $a_i = y$}
\end{array}\right.,
\end{equation}
where $u(\cdot)$ is chosen according to payoff matrix~\eqref{eq:matrix}.
Here, the parameter $p\in(0,\alpha)$ indicates that each agent experiences some small personal cost for selecting the $y$ convention.
Joint actions $\vec{x}:=(x,x,\ldots,x)$ and $\vec{y}:=(y,y,\ldots,y)$, where either all players choose $x$ or all players choose $y$, are Nash equilibria of the game for any graph;
other equilibria may also exist depending on the structure of graph $g$.

Suppose agents in $N$ interact according to the graphical coordination game above, specified by the tuple $(g,\alpha,p)$.
The system operator's objective is that the agents maximize the \emph{system welfare}, or the sum of nominal agent utilities:
\begin{equation} \label{eq:welf}
    W(a) = \sum_{i\in N} U_i(a).
\end{equation}

An adversary seeks to influence the system's emergent behavior by modifying agent utility functions; the adversary accomplishes this by selecting a set of agents and posing as a neighbor whose action is fixed at $y$.%
\footnote{
Note that an adversary could play fixed action $x$ as well, as considered in~\cite{Canty2018,Paarporn2019}, and that situations exist in which a highly intelligent adversary could influence users to select any arbitrary configuration of actions.
For this study, in order to focus on the specific questions of adversary intent, we consider the case that the adversary impersonates only $y$-playing agents and leave the more detailed study of $x$-playing adversaries to future work.
}
Effectively, the adversary spreads ``impostor'' agents throughout the network, attaching each impostor to a specific friendly agent.
We write $K=(i_j)_{j=1}^k$ to represent the adversary's selection, where $i_j$ denotes the index of the friendly agent connected to adversary vertex $j$; we write $\K$ to denote the set of all feasible adversary selections.
That is, each impostor vertex connects to one and only one friendly agent, but note that we allow multiple impostor vertices to connect to a single friendly agent.
Agents' utilities, $\tilde{U}:\mathcal{A}\times \K\to\R$, are now a function of adversarial and friendly behavior, defined by:
\begin{equation}\label{e:new utility}
\tilde{U_i}(a,K) =
U_i(a)	+ \big| \left\{ j : i_j\in K \right\} \big| \mathds{1}_{a_i=y}  . 
\end{equation}
That is, in addition to the ordinary payoff from coordinating with its neighbors, if an agent selects $y$, it receives an additional payoff of $1$ for each impostor node that it is connected to.

Note that the payoffs in this game define an exact potential game~\cite{Monderer1996}.
Given action profile $a$, let $E_x(a)$ ($E_y(a)$) be the set of edges connecting $x$-playing ($y$-playing) agents, and let $N_y(a)$ be the set of $y$-playing agents, and let $N_{Ky}(a)$ be the \emph{number} of total adversarial attacks that are successfully converting an agent to $y$ in $a$.
Then it can be shown that the game's potential function is given by
\begin{equation}
    \Phi(a) = (1+\alpha)|E_x(a)| + |E_y(a)| - pN_y(a) + |N_{Ky}(a)|.
\end{equation}

\subsection{Learning and Emergent Behavior}

Now, suppose that agents in $N$ update their actions over time according to some stochastic process such as the well-studied \emph{log-linear learning} algorithm.
Log-linear learning is known to select a potential function maximizing state with high probability in the long run; it is said that the potential function maximizers of the game are the \emph{stochastically stable states} of log-linear learning.
For reasons of space, we omit the full details, which can be found in~\cite{Alos-Ferrer2010,Brown2017f,Canty2018}.

Accordingly, as we assume the agents are employing log-linear learning and thus selecting potential-maximizing states, Definition~\ref{def:emergent} gives our central solution concept:

\begin{defn} \label{def:emergent}
Given adversarial influence $K$, $\alpha\in(0,1]$, and $p\in(0,\alpha)$, action profile $a\in\aaa$  is called an \emph{emergent state} with respect to $K$, written $\aem(g,\alpha,p,K)$, if for each component of graph $g$, all agents in that component are selecting the same action%
\footnote{
This requirement essentially reflects a situation in which the adversary, upon allocating attacks to a graph component, randomizes the identities of the nodes within that component which are being influenced at each time instance, much like the ``uniformly random'' adversary type in~\cite{Brown2017f}.
For reasons of space, we leave a detailed characterization of this interpretation for future work.
}
(either $x$ or $y$) and that action maximizes the component's potential function.
\end{defn}

Suppose component $c$ with $m_c$ edges and $n_c$ nodes is being attacked by $k_c$ impostor nodes.
Writing $a_c$ to denote the action profile restricted to $c$, we can write the potential function of this component as
\begin{equation}
    \Phi_c(a_c,k_c)=
    \begin{cases} 
       m_c+m_c\alpha & \mbox{if }a_c=\vec{x} \\
       m_c-n_cp+k_c & \mbox{if }a_c=\vec{y} \\
   \end{cases}
\end{equation}

Let  $k^*(c)$ be defined as the minimum $k$ required such that the emergent state of component $c$ is weakly $\vec{y}$.
\begin{equation}
    k^*(c)=\min\{k \mid  \Phi_c(\vec{y},k)\ge\Phi_c(\vec{x},k)\}.
\end{equation}
It can be shown that
\begin{equation}
    k^*(c)=\ceil{m_c\alpha+n_cp}.
\end{equation}

In this paper, we consider two distinct types of adversary, which we term \emph{advertiser} $\A(k)$ and \emph{malicious} $M(k)$.
An advertiser adversary can be thought of as a promoter for the $y$ action (perhaps a marketing manager for a competing product or a political rival), and arranges its attacks in an attempt to maximize the number of agents selecting $y$ in an emergent state.
That is, the advertiser desires to maximize the objective function
\begin{equation} \label{eq:ua}
U^\A(\aem,K;g) = \big| \{ i : \aem_i=y \} \big|.
\end{equation}
A malicious adversary desires to maximize the objective function\vs\vs
\begin{equation} \label{eq:um}
U^\M(\aem,K;g) = -\sum_{i\in N} U_i(\aem),
\end{equation}
equivalently minimizing the system's total welfare.

In this paper, we study a situation in which the system operator selects a network topology $g$ in anticipation of an adversary attack.
Thus, we nominally have a Stackelberg game in which the operator acts as leader and the adversary acts as follower.
For any network topology $g$, we assume that the adversary plays a best response to $g$; for adversary type $T\in\{\A(k),\M(k)\}$, we write 
\begin{equation}
K^T(g) = \argmax_{K\in\K} U^T(\aem,K;g)
\end{equation} to denote the adversary's best response to $g$ at an emergent state (breaking ties in favor of the adversary). 
For convenience, we write 
\begin{equation} \label{eq:W}
W(g,T,\alpha,p) \triangleq W\left(\aem\left(g,\alpha,p,K^{T}(g)\right)\right),
\end{equation} 
or simply $W(g,T)$ when the dependence on $\alpha, p$ is clear.
The system operator's nominal goal is to select a graph topology $g$ which maximizes~\eqref{eq:W}.
For adversary type $T\in\{\A,\M\}$, we write
\begin{equation}
g^*(T) \triangleq \argmax_g W(g,T,\alpha,p).
\end{equation}

For simplicity, we assume throughout that the system contains exactly twice as many agents as graph edges; writing $m:=|E|$, we have that $2m=n$.
We note that $n>2m$ leaves a quantity of nodes that a system designer can never interact with.
Conversely, for $2m<n$ the possible edge topologies are a strict subset of the case when $m=2n$.
We often utilize several specific graph topologies, defined as:
\begin{itemize}
    \item ${\rm SPARSE}$: A graph with exactly $m$ 2-node components and no isolated nodes.
    \item ${\rm LINE}$: A graph with a single sparse giant component described by a line graph containing $m+1$ nodes; this graph contains $m-1$ isolated nodes.
    \item ${\rm COMPLETE}$: A graph with a single dense giant component described graph containing all edges and the minimum-possible number of nodes.
\end{itemize}

In this paper, we study decision making in the face of strategic uncertainty; accordingly, we ask the following: if the system operator is uncertain about the adversary's goal (i.e., type), how should she select graph topology?
To begin to address this question in the context of network design for graphical coordination games, in this paper we study the cost of mischaracterizing the adversary type for a variety of situations.
That is, for $T,T'\in\{\A(k),\M(k)\}$, we define the \emph{relative regret} associated with type $T$ as 
\begin{equation} \label{eq:regret}
R(T) = \frac{W(g^*(T'),T',\alpha,p) -  W(g^*(T),T',\alpha,p)}{W(g^*(T),T',\alpha,p)},
\end{equation}
which captures the fraction of welfare that is lost by planning for type $T$ when the actual realized adversary type is $T'$.

\section{Our Contributions}

We begin with a result which expresses a threshold on $k$ below which the operator can guarantee a completely-coordinated graph.
For adversary type $T\in\{M,A\}$, we define $k_{\rm max}(T(k))$ as the largest value of $k$ \emph{below which} the operator has a graph $g$ such that $\aem(g,\alpha,p,K^T(g))$ is fully coordinated on $x$ actions. 
That is, if $k<k_{\rm max}(T)$, then the operator can be sure that a strategic adversary of type $T$ will not cause any coordination on $y$ in any emergent state, even if isolated nodes are playing $y$. 
Note that the operator always prefers to avoid $y$ coordination if possible (see Lemma~\ref{component slope>lone node slope}).
In the following, $d$ is defined to be the minimum integer that satisfies $m\le\binom{d}{2}$.
We introduce the following technical definition:
\begin{defn}
A triple $(m,\alpha,p)$ is said to be \emph{ordinary} if there exists some graph $g=(N,E)$ such that every component $c$ in $g$ satisfies $n_c<k^*(c)$.
\end{defn}

\begin{theorem}
    \label{kmax lemma}
    We find $k_{\rm max}(\cdot)$ for malicious adversary $M$ and advertiser adversary $A$.
    $$k_{\rm max}(M)=\ceil{m\alpha+mp+p},$$
    \begin{equation*}
    k_{\rm max}(A)\hs=\hs
        \begin{cases} 
           2m-2d\hs+\hs\ceil{m\alpha+dp} & \hspace{-2mm}\mbox{if }(m,\alpha,p)\mbox{ is ordinary} \\
           k_{\rm max}(M) & \hspace{-2mm}\mbox{otherwise.} \\
        \end{cases}
    \end{equation*}
\end{theorem}

\noindent The proof of Theorem~\ref{kmax lemma} appears in the Appendix.

Next, we show that when $k<k_{\rm max}(M)$ (as reported above), the cost of conservatism is low; the optimal defense against a worst-case malicious adversary is asymptotically optimal against an advertiser adversary as well. 

\begin{theorem} \label{low k thm}
    For any feasible $m,\alpha, p$, when $k<k_{\rm max}(M)$, it holds that when playing against a malicious adversary, the LINE graph garners within $2p$ of optimal welfare.
    Furthermore, it holds that
    \begin{equation}
        R(M)= {\cal O}\left({1}/{\sqrt{m}}\right). 
    \end{equation}
    
\end{theorem}
\vspace{2mm}
\begin{proof}
Note that all supporting lemmas appear in the Appendix.
The first result of the Theorem comes directly from Lemma~\ref{optimal malicious}.

We begin deriving the second result by seeking an upper bound on $R(M)$ for $k<k_{max}(M)$.
We present lower bound  on $W(g^*(M(k)),A(k))$ as 
$$W(g^*(M(k)),A(k))\ge W(g^*(M(k)),M(k))$$
$$\ge W({\rm LINE},M(k))=W({\rm LINE},A(k)).$$
Note the last equality holds as $k_{max}(M)$ was derived in Theorem~\ref{kmax lemma} via a LINE graph, thus leaving any adversary with trivial choices while $k<k_{max}(M)$. 
This allows Lemma~\ref{sparsish has trivial benefit} to directly apply to the numerator of $R(M)$ producing upper bound $$R(M)\le \frac{{\cal O}(\sqrt{m})}{W({\rm LINE},M(k))}.$$
When $k<k_{max}(M)$ it can be observed that
$W({\rm LINE},M(k))=2m+2m\alpha-p*\min\{k,m-1\}=\Omega(m).$
Thus 
\begin{equation}
    R(M)={\cal O}\left(\frac{1}{\sqrt{m}}\right)
\end{equation}
as desired.
\end{proof}
\begin{theorem} \label{med k thm}
     When $k\in[k_{\rm max}(M),k_{\rm max}(A)]$, we have that the maximum achievable welfare in the presence of an advertiser is significantly higher than that available in the presence of a malicious adversary:
     \begin{equation}
        \label{raw regret}
         W(g^*(A),A(k))-W(g^*(M)),M(k) \hspace{-.5mm}\geq \hspace{-.5mm} \left\lfloor\hspace{-.5mm}\frac{k}{3}\hspace{-.5mm}\right\rfloor\hspace{-1mm}(2\alpha+2p)-kp. 
     \end{equation}
\end{theorem}
\vspace{2mm}

\begin{proof}
Note that all supporting Lemmas appear in full in the Appendix.
    To compute a lower bound on $\eqref{raw regret}$ we seek a lower bound on 
    \begin{align*}
        & W(g^*(A),A(k))\ge W(\rm{COMPLETE},A(k)) \\
        & \ge 2m+2m\alpha-kp,
    \end{align*}
    and Lemma~\ref{sparse optimal middle k} gives us $$W(g^*(M(k),M))\le2m(1+\alpha)-\left\lfloor{k}/{3}\right\rfloor(2\alpha+2p).$$
    
    Calculating \eqref{raw regret} we get 
    \begin{align*}
        & 2m+2m\alpha-kp-\left(2m+2m\alpha-\left\lfloor{k}/{3}\right\rfloor(2\alpha+2p)\right) \\
        & = \left\lfloor{k}/{3}\right\rfloor(2\alpha+2p)-kp.
    \end{align*}
\end{proof}

\section{Simplified Example} \label{simple section}
In this section we will restrict the system planner to choosing between the three canonical graph designs, SPARSE, COMPLETE, and LINE.
In this environment we completely characterize optimal decision making.

\begin{figure}
\centerline{\includegraphics[scale=0.6]{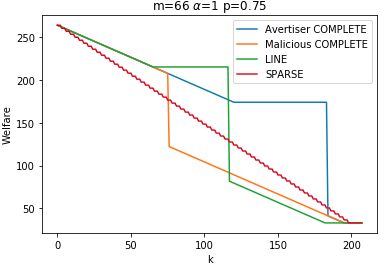}}
\caption{We demonstrate that $k_{max}(M)=117$ and $k_{max}(A)=183$, realized by LINE and COMPLETE respectively, serve as important breakpoints in optimal decision making. In this case $\alpha+2p>2$, so SPARSE dominates for either adversary after their respective $k_{max}$ value.}
\label{fig:a+2p>2}
\vs\vs\vs\vs\vs\vs
\end{figure}

\subsection{Advertiser Characterization}
We may observe Theorem~\ref{kmax lemma}'s $k_{max}$ values were realized using LINE and COMPLETE graphs and Lemma~\ref{component slope>lone node slope} provides meaning for the $k_{max}$ regions as relatively high welfare zones. 
So if $k<k_{max}(M)$ we may readily play LINE.
Now using $k_r$ from the proof of Lemma~\ref{sparse optimal middle k} we observe two cases on $\alpha,p$. They are
 $\alpha+2p>2$ which leads to
\begin{equation}
    g^*(A(K))=
    \begin{cases} 
       {\rm LINE} & k<k_{max}(M) \\
       {\rm COMPLETE} & k_{max}(M)\le k<k_{max}(A)\\
       {\rm SPARSE} & else\\
   \end{cases}
\end{equation}
or alternatively $\alpha+2p\le2$
\begin{equation}
    g^*(A(K))=
    \begin{cases} 
       {\rm LINE} & k<k_{max}(M) \\
       {\rm COMPLETE} & k_{max}(M)\le k.\\
   \end{cases}
\end{equation}


\subsection{Malicious Characterization}
We once again are able to readily play LINE if $k<k_{max}(M)$ by similar arguments.
However once $k$ exceeds $k_{max}(M)$ we must chose between potentially all 3 graphs.

It can be shown that $k_r(COMPLETE)\geq k_r(LINE)$ and by linear arguments COMPLETE will at least weakly dominant LINE for all $k\geq k_{max}(M)$. 
Thus we will define $k_{int}$ as the intersection between SPARSE and COMPLETE and use it as decision point between the two.

\begin{figure}
\centerline{\includegraphics[scale=0.6]{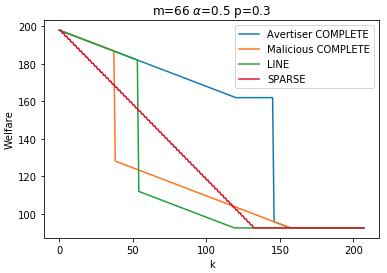}}
\caption{In addition to using $k_{max}$ for optimal decision making, we must compute a $k_{int}$ representing the intersection between SPARSE and LINE because $\alpha+2p\le2$.}
\label{fig:a+2p<2}
\vs\vs\vs\vs\vs\vs
\end{figure}

For our first case where $\alpha+2p>2$ we find
\begin{equation}
    g^*(A(M))=
    \begin{cases} 
       {\rm LINE} & k<k_{max}(M) \\
       {\rm SPARSE} & k_{max}(M)\ge k.\\
   \end{cases}
\end{equation}
If $\alpha+2p\le 2$ we simply use $k_{int}$ as a break point where SPARSE switches to COMPLETE.
\begin{equation}
    g^*(A(M))=
    \begin{cases} 
       {\rm LINE} & k<k_{max}(M) \\
       {\rm SPARSE} & k_{max}(M)\le k<k_{int}\\
       {\rm COMPLETE} & k_{int} \le k.
   \end{cases}
\end{equation}
We showcase the results with sample parameters for $\alpha+2p>2$ in Figure~\ref{fig:a+2p>2} and $\alpha+2p\leq2$ in Figure~\ref{fig:a+2p<2}.

\bibliographystyle{ieeetr}
\bibliography{library}

\appendix
\subsection{Notation}


For component $c$ in some graph we define
$\Delta W(c) := 2m_c\alpha+n_cp$
as the loss in welfare when component $c$ switches from being stable on $x$ to $y$.

We define several sets
\begin{itemize}
    \item $C(g)=\{c_1,c_2,c_3,\dots,c_n\}$ such that $c_i\in C(g)$ is a component in graph g
    \item $C_e(g)=\{c_1,c_2,c_3,\dots,c_n\}$ such that $c_i\in C_e(g)$ is an edged component in graph g
    \item $C_l(g)=\{c_1,c_2,c_3,\dots,c_n\}$ such that $c_i\in C_l(g)$ is a node with degree 0
    \item $C_x(g|T(k))=\{c_1,c_2,c_3,\dots,c_n\}$ such that $c_i\in C_x(g\mid T(k))$ is a component with emergent state $\vec{x}$
    \item $C_y(g|T(k))=\{c_1,c_2,c_3,\dots,c_n\}$ such that $c_i\in C_y(g\mid T(k))$ is a component with emergent state $\vec{y}$.
\end{itemize}

We use 
\begin{equation}
    W(g,T(K))=2m+2m\alpha-\sum_{c\in C_y(g,A(k))}\Delta W(c)
\end{equation}
as a convenient way to denote system welfare given topology and the adversary.

We term some component $c$ \emph{advertiser resistant} if $c$ satisfies $n_c<k^*(c)$.
\subsection{Proof of Theorem~\ref{kmax lemma}}
    \label{kmax proof}
    First we derive $k_{max}(M)$.
    Let $g$ be an arbitrary graph and $q\in\argmin_{c\in C_e(g)}k^*(c)$.
    Lemma~\ref{component slope>lone node slope} gives the  malicious adversary strictly prefers $C_y(g|M(k))\subseteq C_e(g)$.
    Therefore we can simply seek to design $q$ to maximize $k^*(q)$.
    First,
    \vs\vs
    \begin{equation}
        k_{max}(M)=\max_{n_q,m_q} \ceil{m_q\alpha+n_qp}=\ceil{m\alpha+(m+1)p}\vs
    \end{equation}
    simply by problem parameters.
    Now we seek $k_{max}(A)$.
    We begin with the case where in the triple $(m,\alpha,p)$ is not ordinary.
    This implies all graphs contain a nonempty set $S$ of non advertiser resistant components and we take $c=\argmin_{s\in S} k^*(s)$.
    Thus at $k=k^*(c)$, $n_c\geq k$ and intuitively the advertiser may weakly improve $U^a$ by taking component $c$ over lone nodes. This behavior is similar to the malicious adversary and we conclude $k_{max}(A)=k_{max}(M)$ in this case.
    
    For brevity we present an outline of the case where $(m,\alpha,p)$ is ordinary.
    \begin{itemize}
        \item Given $(m,\alpha,p)$ is ordinary, let all components in graph $g$ be advertiser resistant
        \item If there are plentiful lone nodes, advertiser resistance gives the advertiser prefers to attack the lone nodes.
        \item Once all lone nodes are stable on $y$ the advertiser may gather excess attacks to make up the difference in $U^a(\cdot)$ that occurs taking some non-trivial component.
        \item It can be shown that the minimal $k$ where this happens may be weakly increased by removing two components in $g$ and creating a new component with fewer nodes.
        \item We finish optimizing the graph with a single non-trivial component by making making it as near complete as possible. \hfill\QED
    \end{itemize}
\begin{lemma}
    \label{sparse optimal middle k}
We introduce an upper bound of $$W(g^*(M),M)\le2m(1+\alpha)-\left\lfloor\frac{k}{3}\right\rfloor(2\alpha+2p)$$ for $k\in[k_{max}(M), k_{max}(A))$. 
\end{lemma}
\begin{proof}
    Observe the cost to take every component in the graph is given by
    \begin{equation}\label{e:original utility}
        k_{r}(g)=\sum_{c\in C(g)}k^*(c)
    \end{equation}
    Noting $\max_{\alpha,p}(k^*(c))=n_c+m_c$ and
    \begin{equation}
        \sum_{c\in C(g)}m_c=m \And \sum_{c\in C}n_c=2m
    \end{equation}
    by problem parameters we may obtain $\max_{\alpha,p,g}k_{r}(g)=3m$ across all possible graphs.
    
    Now we examine a SPARSE design with $\alpha+2p>2$.
    We compute $k^*(c)=\ceil{\alpha+2p}=3$ for $\forall c\in C_e({\rm SPARSE})$ so we obtain $k_r({\rm SPARSE})=3m$.
    
    Let $g'$ be some graph such that $g'\neq{\rm SPARSE}$.
    $C_l(g')$ must then be nonempty and by Theorem~\ref{kmax lemma} we know for $\exists c\in C_e(g')$ such that $c\in C_y(g)$.
    Thus $g'$ features two unique $\frac{\Delta W(c)}{k^*(c)}$ for $c\in C_l(g')$ and $c \in C_e(g')$.
    By Lemma~\ref{component slope>lone node slope} we know $\frac{\Delta W(s)}{k^*(s)}>k^*(c)p$ for $s\in C_e(SPARSE)$.
    Knowing $W(g',M(0))=W({\rm SPARSE},M(0))$ and $W(g',M(3m))=W({\rm SPARSE},M(3m))$, we deduce then that $\frac{W(c)}{k^*(c)}>\frac{W(s)}{k^*(s)}$ for $c\in C_e(g')$.
    Thus $W({\rm SPARSE},M(k))\ge W(g',M(k))$ for $k\ge k_{max}(M) \And \forall g'$, thus making SPARSE optimal.
    
    Now we define SPARSE$^*$ to be a fictional graph identical with SPARSE except it satisfies $k^*(c)=3$ for $c\in C_e(SPARSE^*)$ for all $\alpha,p$.
    So $k_r({\rm SPARSE}^*)=3m$ and we may then follow similar arguments that $k_r(g')$ is maximally $3m$ and $g'$ must contain some component satisfying $\frac{W(c)}{k^*(c)}>\frac{W(s)}{k^*(s)}$ for $c\in C_e(g'), s\in C_e({\rm SPARSE}^*)$.
    Thus $W({\rm SPARSE}^*,M)\ge W(g',M)$ by the same arguments as sparse $\forall g',\alpha,p$.
    
    So we may present the welfare function of ${\rm SPARSE}^*$ as
    \begin{equation}
        W(g^*(M),M(k))\le2m(1+\alpha)-\left\lfloor{k}/{3}\right\rfloor(2\alpha+2p).\vs\vs
    \end{equation}
\end{proof}

\begin{lemma}
    \label{sparsish has trivial benefit}
    When $k<k_{\rm max}(M)$, the welfare benefit of $g^*(A(k))$ over LINE is ${\cal O}(\sqrt{m})$: 
    \begin{align}
        B(k)    &\triangleq W\left(g^*(A),A(k)\right) \hs-\hs W({\rm LINE},A(k)) = {\cal O}(\sqrt{m}) 
    \end{align}
\end{lemma}

\begin{proof}
    First we may appreciate the implications of Lemma \ref{component slope>lone node slope}
    such that any candidate optimal graph $g$ must satisfy $C_e(g)\cap C_y(g|A(k))=\emptyset$.
    
    And by Theorem~\ref{kmax lemma}'s $k_{max}$ values were obtained via single nontrivial component graphs.
    However even the most sparse of these designs, the line graph, leaves $(m-1)$ nodes unconnected.
    That is, 
    \begin{equation} \label{eq:welf-line}
    W\left({\rm LINE},A(k)\right) = 2m(1+\alpha) - p\min\{k,m-1\},
    \end{equation}
    So if we can satisfy the first requirement with a many-tree component graph, we could increase welfare.
    
    
    
    \addtolength{\textheight}{-1mm}
    
    Observe that any tree component will cover $m_c+1$ nodes and any component with at least one cycle will cover strictly less. 
    Therefore we can let $n_c=m_c+1$.
    Note that $\sum_{c\in C_e(g)}m_c+1=m+|C_e(g)|$ thus forest graph $g$ contains $|C_l(g)|-|C_l(LINE)|=|C_e(g)|-1$ less lone nodes.
    Thus, an optimal graph $g$ in this regime is composed entirely of tree components and lose none to an advertiser.
    That is, it holds that
    \begin{equation} \label{eq:welf-ad-opt}
        W\left(g,A(k)\right) \leq 2m(1+\alpha) - p\min\{k,m-1\} + p(|C_e(g)|-1).
    \end{equation}
    This is due to the fact that each additional component protects at most one additional node, potentially improving welfare relative to~\eqref{eq:welf-line} by an amount $p$ for some $k$.
    Combining~\eqref{eq:welf-line} and~\eqref{eq:welf-ad-opt}, we obtain upper bound.
    \begin{equation}
        B(k) \leq (|C_e(g)|-1)p. \label{eq:bound1}
    \end{equation}
    
    
   
    Consider a situation such that forest graph $g^*$ has $|C_e(g^*)|$ non trivial components and $m-|C_e(g^*)|$ lone nodes.
    Let $c\in \argmin_{q\in C_e(g^*)}k^*(q)$.
    When $m-|C_e(g^*)|<k<m-|C_e(g^*)|+\ceil{m_c\alpha+n_cp}$, suppose the advertiser takes all lone nodes but no component.
    Then we have
    \begin{align}
        B(k)    &=      2m(1+\alpha)-(m-|C_e(g^*)|)p-(2m(1+\alpha)-kp) \nonumber\\
                &\leq   (m-|C_e(g^*)|+\ceil{m_c\alpha+n_cp} - m +|C_e(g^*)|)p \nonumber \\
                &= \ceil{m_c\alpha+n_cp}p. \label{eq:bound2}
    \end{align}
    
    Combining~\eqref{eq:bound1} and~\eqref{eq:bound2} and noting $g^*$ is a forest, we have
    \begin{equation} \label{eq:Bbound}
        B(k) \leq p\cdot\min\{(|C_e(g^*)|-1),\ceil{m_c(\alpha+p) + p}\}.
    \end{equation}
    Recall that $|C_e(g^*)|\leq\floor{\frac{m}{m_c}}$. 
    Then if $m_c>\sqrt{m}$, 
    \begin{align}
        |C_e|     \leq      \left\lfloor{m}/{m_c}\right\rfloor 
                \leq   \left\lfloor{m}/{\sqrt{m}}\right\rfloor  
                = \floor{\sqrt{m}}. \label{eq:cbound}
    \end{align}
    That is, either  $m_c\leq\sqrt{m}$ or $|C_e(g^*)|\leq\sqrt{m}.$
    Thus it is clear from~\eqref{eq:Bbound} that
    \begin{equation*}
        B(k)={\cal O}(\sqrt{m})
    \end{equation*}
    as desired.
\end{proof}

\begin{lemma}
    \label{component slope>lone node slope}
    If $c \in C_{\rm e}(g)$ and $c' \in C_l(g)$ then $\frac{\Delta W(c)}{k^*(c)}>\frac{\Delta W(c')}{k^*(c')}$ or equivalently $\Delta W(c)>k^*(c)p$. Thus we can guarantee system designers have a strong preference that $C_y(g\mid T(k)) \subseteq C_l(g)$.
\end{lemma}
\begin{proof}
    First note that
    
    ${\Delta W(c')}/{k^*(c')}={p}/{\ceil{0\cdot\alpha+p}}=p,\vspace{0mm}$
    \vspace{-1mm}
    and
    $$\frac{\Delta W(c)}{k^*(c)}=\frac{2m_c+2m_c\alpha-(2m_c-n_cp)}{\ceil{m_c\alpha-n_cp}}=\frac{2m_c\alpha+n_cp}{\ceil{m_c\alpha-n_cp}}.$$
    Thus, we wish to show 
    $p<\frac{2m_c\alpha+n_cp}{\ceil{m_c\alpha-n_cp}}$,
    or equivalently
    $\ceil{m_c\alpha-n_cp}p<2m_c\alpha+n_cp$.
    
    Let $R=\ceil{m_c\alpha-n_cp}-m_c\alpha-n_cp$ such that $R\in [0,1)$.
    So we may obtain $m_c\alpha p+n_cp^2+Rp<2m_c\alpha+n_cp$, which may be simplified further into $Rp<m_c\alpha+n_cp(1-p)+m_c\alpha(1-p)$.
    We may further observe that $m_c\alpha+n_cp(1-p)+m_c\alpha(1-p)\ge\alpha$ and clearly $p\ge Rp$.
    Thus $p<\alpha \implies Rp<m_c\alpha+n_cp(1-p)+m_c\alpha(1-p)$.
\end{proof}

\begin{lemma}
    \label{optimal malicious}
    When $k<k_{\rm max}(M),$ a line graph is within $2p$ of optimal against a malicious adversary for $k<k_{max}(M)$:
    \begin{equation} \label{eq:line optimal}
         W(g^*(M(k)),M(k)) - W({\rm LINE},M(k)) \leq 2p.
    \end{equation}
    and it holds that
    \begin{equation}
        W(g^*(M(k)),M(k)) \geq 2m(1+\alpha) - p\min\{k,m-1\}.
    \end{equation}
\end{lemma}

\begin{proof}
    First we seek conditions on arbitrary graph $g$ such that it may satisfy 
    \begin{equation}
        \label{cond:better then line}
        W(g|M(k))>W(LINE|M(k))
    \end{equation}
    for some $k\in [0,k_{\rm max}(M))$.
    Theorem \ref{kmax lemma} and Lemma \ref{component slope>lone node slope} give that for
     $q\in\argmin_{c\in C_e(g)} k^*(c)$
    \begin{equation}
        \label{cond:component q condition}
        k^*(q)>k
    \end{equation}
     must hold or (\ref{cond:better then line}) will not be satisfied.
     
    We may also deduce that to satisfy (\ref{cond:better then line}) there must exist some $k$ where $|C_y(g|M(k))|<|C_y({\rm LINE}|M(k))|$ knowing $C_y(g|M(k))\subseteq C_l(g|M(k))$, thus we can conclude the following two conditions must hold for $g$ to satisfy (\ref{cond:better then line}):
    \vs
    \begin{equation}
        \label{cond: protect more nodes}
        |C_l(g)|<|C_l({\rm LINE})|=m-1 \mbox{ and}\vs\vs\vs
    \end{equation}
    \vs
    \begin{equation}
    \vs
        \label{cond:kc condition}
        k>|C_l(g)|.
    \end{equation}
    That is, graph $g$ has fewer lone nodes then LINE, so there may exist a $k$ where $g$ will feature less $\vec{y}$ lone nodes, and $k$ must be large enough to be in this zone.
    Let $k_g$ be the minimal such $k$ to satisfy (\ref{cond:kc condition}) on graph $g$.
    
    Now we note that $k^*(q)$ may be increased and $k_g$ may be decreased by following a simple algorithm where all components are destroyed and replaced with tree components with corresponding edge counts. 
    So we seek forest graphs satisfying 
    $k^*(q)>k_g$.

    On forest graphs we may easily compute $k_g=2m-m-|C_e(g)|+1=m-|C_e(g)|+1$ as tree components always satisfy $n_c=m_c+1$. 
    We may once again optimize forest graph $g$ then by removing all components except $q$ and and producing as many copies of $q$ as possible.
    However, the process will not be feasible if $m_q \nmid m$, so we use $\frac{m}{m_q}$ as an upper bound on $|C_e(g)|$ while maintaining a given $k^*(q)$ thus producing lower bound
%
    \begin{equation}
        \label{bound:k_c lowerbound}
        \underbar{k}_g:=m-{m}/{m_q}+1=m({m_q-1})/{m_q}+1\leq k_g.
    \end{equation}
    Since \eqref{bound:k_c lowerbound} is independent of $\alpha,p$, we take $a=p=1$ to deduce upper bound
    $\bar{k}^*(q)=2m_q+1\geq\ceil{m_q\alpha+n_qp}=k^*(q).$
    We now show $k_g\geq k^*(q)$ holds for $m_q\in\{2,3,4,...,\floor{\frac{m}{3}}\}$ by showing $\underbar{k}_g\geq\bar{k}^*(q)$.
    It follows algebraically that
    \begin{equation}
        \underbar{k}_g\geq\bar{k}^*(q)\Leftrightarrow m\ge {2m_q^2}/({m_q-1}),
    \end{equation}
    and the right hand statement is always true since
    \begin{align}
        2\frac{m_q^2}{m_q-1}\leq2\frac{m_q^2+m_q-2}{m_q-1}=2m_q+4\leq m, \label{eq:random}
    \end{align}
    where the inequalities follow from $2\leq m_q\leq m/3$ and $m\geq12$.
    Thus (\ref{cond:component q condition}) never holds for graphs satisfying $2\leq m_q\leq m/3$.

    If $m_q=1$ (i.e., $g={\rm SPARSE}$), it can be shown that $W({\rm SPARSE},M(K))-W({\rm LINE},M(k))\le2p$.
    If $m_q>m/3$, the resulting $g$ can have at most $2$ nontrivial components, and it can be shown that $W(g,M(K))-W({\rm LINE},M(k))\le p.$
\end{proof}       

\end{document}